\documentclass{siamart190516}

\usepackage{amsmath,amsfonts}
\usepackage{xspace}

\headers{Dynamic Balanced Graph Partitioning}{C.~Avin, M.~Bienkowski, A.~Loukas, M.~Pacut, and S.~Schmid}

\title{Dynamic Balanced Graph Partitioning\thanks{%
A preliminary version of this paper appeared as 
``Online Balanced Repartitioning'' in the proceedings of the
30th International Symposium on DIStributed Computing (DISC 2016).
\funding{Research supported by the German-Israeli Foundation for Scientific Research
(GIF) Grant I-1245-407.6/2014, the Polish National Science Centre grants
2016/22/E/ST6/00499 and 2016/23/N/ST6/03412, and ERC Consolidator project AdjustNet grant agreement No.~864228.}}}

\author{
Chen~Avin\thanks{School of Electrical and Computer Engineering, Ben Gurion University of the Negev, Israel} \and
Marcin~Bienkowski\thanks{Institute of Computer Science, University of Wroclaw, Poland} \and
  Andreas~Loukas\thanks{EPFL, Switzerland} \and
  Maciej~Pacut\thanks{Faculty of Computer Science, University of Vienna, Austria} \and
  Stefan~Schmid\footnotemark[5]
}

\ifpdf
\hypersetup{
  pdftitle={Dynamic Balanced Graph Partitioning},
  pdfauthor={C.~Avin, M.~Bienkowski, A.~Loukas, M.~Pacut, and S.~Schmid}
}
\fi

\newcommand{\ONL}{\textsc{Onl}\xspace}
\newcommand{\ONLreq}{\textsc{Onl}^\textnormal{req}\xspace}
\newcommand{\ONLmig}{\textsc{Onl}^\textnormal{mig}\xspace}
\newcommand{\GREEDY}{\textsc{Greedy}\xspace}
\newcommand{\MGREEDY}{M^\textsc{GR}}
\newcommand{\MOPT}{M^\textsc{OPT}}
\newcommand{\OFF}{\textsc{Off}\xspace}
\newcommand{\OPT}{\textsc{Opt}\xspace}
\newcommand{\ALG}{\textsc{Alg}\xspace}
\newcommand{\CREP}{\textsc{Crep}\xspace}
\newcommand{\DET}{\textsc{Det}\xspace}
\newcommand{\CREPreq}{\CREP^\textnormal{req}}
\newcommand{\CREPmig}{\CREP^\textnormal{mig}}
\newcommand{\ovr}{\textsc{ovr}}
\newcommand{\T}{\mathcal{T}}
\newcommand{\eps}{\ensuremath{\epsilon}}

\newcommand{\cut}{\textsc{cut}}
\newcommand{\eold}{e^\textrm{old}}
\newcommand{\enew}{e^\textrm{new}}

\newcommand{\optmig}{\textsc{opt-mig}}
\newcommand{\spl}{\textsc{sp}}
\newcommand{\final}{\textsc{fin}}
\newcommand{\size}{\textsc{size}}
\newcommand{\comm}{\textsc{comm}}
\newcommand{\Phiinit}{\Phi^\textrm{S}}
\newcommand{\winit}{w^\textrm{S}}
\newcommand{\set}{\mathcal{S}}
\newcommand{\A}{\mathcal{A}}
\newcommand{\X}{\mathcal{X}}
\newcommand{\F}{\mathcal{F}}

\begin{document}

\maketitle

\begin{abstract}
This paper initiates the study of the classic balanced graph partitioning
problem from an online perspective: Given an~arbitrary sequence of pairwise
communication requests between~$n$~nodes, with patterns that may change over
time, the objective is~to~service these requests efficiently by partitioning
the nodes into~$\ell$ clusters, each of size~$k$, such that frequently
communicating nodes are located in the same cluster. The partitioning can be
updated dynamically by \emph{migrating} nodes between clusters. The goal is to
devise online algorithms which jointly minimize the amount of inter-cluster
communication and migration cost.

The problem features interesting connections to other well-known online
problems. For example, scenarios with~$\ell=2$ generalize online paging, and
scenarios with~$k=2$ constitute a~novel online variant of maximum matching. We
present several lower bounds and algorithms for settings both with and without
cluster-size augmentation. In particular, we prove that any deterministic
online algorithm has a competitive ratio of at least~$k$, even with
\emph{significant} augmentation. Our main algorithmic contributions are
an~$O(k \log{k})$-competitive deterministic algorithm for the general setting
with constant augmentation, and a constant competitive algorithm for the
maximum matching variant.
\end{abstract}




\section{Introduction}

Graph partitioning problems, like minimum graph bisection or minimum balanced
cuts, are among the most fundamental problems in theoretical computer science.
They are intensively studied also due to their numerous practical
applications, e.g., in communication networks, parallel processing, data
mining and community discovery in social networks. Interestingly however, not
much is known today about how~to~\emph{dynamically} partition nodes that
interact or communicate in a~time-varying fashion.

This paper initiates the study of a natural model for \emph{online graph
partitioning}. We are given a~set of~$n$ nodes with time-varying pairwise
communication patterns, which have to be partitioned into~$\ell$~clusters of
equal size~$k$. Intuitively, we would like to minimize inter-cluster
interactions by mapping frequently communicating nodes to the same cluster.
Since communication patterns change over time, partitions should be
readjusted dynamically, that is, the nodes should be \emph{repartitioned}, in
an online manner, by \emph{migrating} them between clusters. The objective is
to jointly minimize inter-cluster communication and repartitioning costs,
defined respectively as the number of communication requests ``served
remotely'' and the number of times nodes are migrated from one cluster to
another.

This fundamental online optimization problem has many applications. For
example, in the context of~cloud computing, $n$ may represent virtual machines
or containers that are distributed across~$\ell$ physical servers, each having
$k$ cores: each server can host $k$ virtual machines. We would like to
(dynamically) distribute the virtual machines across the servers, so that
datacenter traffic and migration costs are minimized.

\subsection{The Model}

Formally, the online \emph{Balanced RePartitioning} problem (BRP) is defined as
follows. There is a set of $n$ nodes, initially distributed arbitrarily
across $\ell$~clusters, each of size~$k$. We call two nodes~$u,v\in V$
\emph{collocated} if they are in the same cluster.

An input to the problem is a sequence of communication requests $\sigma =
(u_1,v_1),$ $(u_2,v_2), (u_3,v_3), \ldots$, where pair $(u_t,v_t)$ means that
the nodes $u_t,v_t$ exchange a fixed amount of data. For succinctness of later descriptions,
we assume that a request $(u_t,v_t)$ occurs at time $t \geq 1$. At any time~$t
\geq 1$, an online algorithm needs to serve the~communication
request~$(u_t,v_t)$. Right before serving the request, the online algorithm
can \emph{re}partition the nodes into new clusters. We assume that
a~communication request between two collocated nodes costs 0. The cost of a
communication request between two nodes located in different clusters is
normalized to~1, and the cost of migrating a node from one cluster to another
is~$\alpha \geq 1$, where $\alpha$ is a parameter (an integer). For any
algorithm \ALG, we denote its total cost (consisting of communication plus
migration costs) on sequence $\sigma$ by $\ALG(\sigma)$.

The description of some algorithms (in particular the ones in \cref{sec:upper}
and \cref{sec:crep}) is more natural if they first serve a request and then
optionally migrate. Clearly, this modification can be implemented at no extra cost by
postponing the migration to the next step.

We are in the realm of competitive worst-case analysis and compare the
performance of an online algorithm to the performance of an optimal offline
algorithm. Formally, let~$\ONL(\sigma)$, resp.~$\OPT(\sigma)$, be the cost
incurred by an online algorithm \ONL, resp.~by an optimal offline
algorithm \OPT, for a given~$\sigma$. In contrast to \ONL, which learns the~requests one-by-one as
it serves them, \OPT has a complete knowledge of the entire request
sequence~$\sigma$ \emph{ahead of~time}. The goal is to design online repartitioning
algorithms that provide worst-case guarantees. In particular, $\ONL$ is said
to be \emph{$\rho$-competitive} if there is a constant $\beta$, such that for any
input sequence~$\sigma$ it holds that
\[
	\ONL(\sigma) \leq \rho \cdot \OPT(\sigma) + \beta.
\]
Note that $\beta$ cannot depend on input $\sigma$ but can depend on other
parameters of the problem, such as the number of nodes or the number of clusters.
The minimum $\rho$ for which $\ONL$ is $\rho$-competitive is called the 
\emph{competitive ratio} of $\ONL$. 

We consider two different settings:

\begin{description}

\item[Without augmentation:] The nodes fit perfectly into the clusters,
i.e.,~$n=k\cdot \ell$. Note that in this setting, due to cluster capacity
constraints, a node can never be migrated alone, but it must be \emph{swapped}
with another node at a cost of~$2 \cdot \alpha$. We also assume that when an
algorithm wants to migrate more than two nodes, this has to be done using
several swaps, each involving two nodes.

\item[With augmentation:] An online algorithm has access to additional space
in each cluster. We say that an algorithm is~$\delta$-augmented if the size of
each cluster is~$k' = \delta \cdot k$, whereas the total number of nodes
remains~$n = k\cdot \ell < k'\cdot \ell$. As usual in competitive analysis,
the augmented online algorithm is compared to the optimal offline algorithm
with cluster capacity~$k$.
\end{description}

An online repartitioning algorithm has to cope with the following issues:

\begin{description}

\item[Serve remotely or migrate (``rent or buy'')?] For just a brief communication, it may not be worthwhile to collocate the nodes: the migration cost might
be too large in comparison to communication costs.

\item[Where to migrate, and what?]
If an algorithm decides to collocate nodes $x$ and~$y$, the question becomes
how. Should $x$ be migrated to the cluster holding $y$, $y$~to the one holding
$x$, or should both nodes be migrated to a new cluster?

\item[Which nodes to evict?]
There may not exist sufficient space in the desired destination cluster. In
this case, the~algorithm needs to decide which nodes to ``evict'' (migrate to
other clusters), to free up space.

\end{description}

\subsection{Our Contributions}

This paper introduces the online Balanced RePartitioning problem (BRP),
a fundamental \emph{dynamic} variant of the classic graph clustering problem. 
We show that BRP features some interesting connections to other well-known
online graph problems. For $\ell=2$, BRP can simulate the online paging problem,
and for $k=2$, BRP is a~novel online version of maximum matching.
We consider deterministic algorithms and make the following technical
contributions:

\begin{description}

\item[Algorithms for General Variant:]
For the non-augmented variant, in \cref{sec:upper}, we first present a~simple
$O(k^2 \cdot \ell^2)$-competitive algorithm. Our main technical contribution
is an $O((1+1/\eps) \cdot k \log{k})$-competitive deterministic algorithm
$\CREP$ for a setting with $(2+\eps)$-augmentation (\cref{sec:crep}).
We emphasize that this bound does not depend on~$\ell$. This is interesting,
as in many application domains of this problem, $k$ is small: for example, in
our motivating virtual machine collocation problem, a server typically hosts
only a small number of virtual machines (e.g., related to the constant number
of cores on the server).

\item[Algorithms for Online Rematching:]
For the special case of online rematching ($k=2$, but arbitrary~$\ell$), in
\cref{sec:k-two}, we prove that a variant of a greedy algorithm is
7-competitive. We also demonstrate a lower bound of 3 for any deterministic
algorithm.

\item[Lower Bounds:]
By a reduction to online paging, in \cref{sec:paging}, we show that
for two clusters, no deterministic algorithm can obtain a better bound than
$k-1$. While this shows an~interesting link between BRP and paging, in
\cref{sec:lower-bounds}, we present a stronger bound. Namely, we
show that for ${\ell \geq 2}$ clusters, no deterministic algorithm can beat the
bound of $k$ even with an~arbitrary amount of augmentation, as~long~as~the
algorithm cannot keep all nodes in a~single cluster. In contrast, online
paging is known to~become constant-competitive with constant
augmentation~\cite{SleTar85}.

\end{description}

\subsection{A Practical Motivation}

There are many applications to the dynamic graph clustering problem.
To give just one example, we consider server virtualization in
datacenters. Distributed cloud applications, including batch processing
applications such as MapReduce, streaming applications such as Apache Flink or
Apache Spark, and scale-out databases and key-value stores such as Cassandra,
generate a~significant amount of network traffic and a considerable fraction
of their runtime is due to network activity~\cite{MogPop12}. For example,
traces of jobs from a Facebook cluster reveal that network transfers on
average account for 33\% of the execution time~\cite{ChZMJS11}. In such
applications, it is desirable that frequently communicating virtual machines
are \emph{collocated}, i.e., mapped to the same physical server: 
communication across the network (i.e., inter-server communication) induces
network load and latency. However, migrating virtual machines between servers
also comes at a price: the state transfer is bandwidth intensive, and may even
lead to short service interruptions. Therefore the goal is to design online
algorithms that find a good trade-off between the inter-server communication
cost and the migration cost.


\section{Related Work}\label{sec:relwork}

The static offline version of our problem, i.e., a problem variant where
migration is not allowed, where all requests are known in advance, and where
the goal is to find best node assignment to $\ell$ clusters, is known as the
$\ell$-balanced graph partitioning problem. The problem is 
NP-complete, and cannot even be approximated within any finite factor unless P
= NP~\cite{AndRae06}. The static variant where $n/\ell = 2$ corresponds to a
maximum matching problem, which is polynomial-time solvable. The static
variant where $\ell = 2$ corresponds to the minimum bisection problem, which
is already NP-hard~\cite{GaJoSt76}. Its approximation was studied in a long
line of work~\cite{SarVaz95,ArKaKa99,FeKrNi00,FeiKra02,KraFei06,Raec08} and
the current best approximation ratio of $O(\log n)$ was given by
R{\"{a}}cke~\cite{Raec08}. The $O(\log^{3/2} n)$-approximation given by
Krauthgamer and Feige~\cite{KraFei06} can be extended to~general $\ell$, but
the running time becomes exponential in~$\ell$.

The inapproximability of the static variant for general values of $\ell$
motivated research on the bicriteria variant, which can be seen as the offline
counterpart of our cluster-size augmentation approach. Here, the~goal
is~to~develop $(\ell,\delta)$-balanced graph partitioning, where the graph has
to be partitioned into $\ell$ components of~size less than $\delta \cdot (n /
\ell)$ and the cost of the cut is compared to the optimal (non-augmented)
solution where all components are of size $n / \ell$. The variant where
$\delta \geq 2$ was considered in
\cite{LeMaTr90,SimTen97,EvNaRS00,EvNaRS99,KrNaSc09}. So far the best result is
an $O(\!\sqrt{\log n \cdot \log \ell})$-approximation by Krauthgamer et
al.~\cite{KrNaSc09}, which builds on ideas from the $O(\!\sqrt{\log
n})$-approximation algorithm for balanced cuts by Arora et al.~\cite{ArRaVa09}.
For smaller values of $\delta$, i.e., when $\delta = 1 + \eps$ with a fixed
$\eps > 0$, Andreev and R{\"{a}}cke gave an $O(\log^{1.5} n / \eps^2)$
approximation~\cite{AndRae06}, which was later improved to $O(\log n)$ by
Feldmann and Foschini ~\cite{FelFos15}.

The BRP problem considered in this paper was not previously studied. However,
it bears some resemblance to the classic online problems; below we highlight
some of them.

Our model is related to online
paging~\cite{SleTar85,FKLMSY91,McGSle91,AcChNo00}, sometimes also referred to
as online caching, where requests for data items (nodes) arrive over time and
need to be served from a cache of finite capacity, and where the number of
cache misses must be minimized. Classic problem variants usually boil down to
finding a smart eviction strategy, such as Least Recently Used (LRU). In our
setting, requests can be served remotely (i.e., without fetching the
corresponding nodes to a single cluster). In this light, our model is more
reminiscent of caching models \emph{with
bypassing}~\cite{EpImLN11,EpImLN15,Irani02}. Nonetheless, we show that BRP is
capable of emulating online paging.

The BRP problem is an example of a non-uniform problem~\cite{KaMaMO94}: the
cost of changing the state is higher than the cost of serving a single
request. This requires finding a~good trade-off between serving requests
remotely (at a low but repeated communication cost) or migrating nodes into a
single cluster (entailing a potentially high one-time cost). Many
online problems exhibit this so called \emph{rent-or-buy} property, e.g., ski
rental problem~\cite{KaMaMO94,LoPaRa08}, relaxed metrical task
systems~\cite{BaChIn01}, file migration~\cite{BaChIn01,BiByMu17}, distributed
data management~\cite{BaFiRa95,AwBaFi93,AwBaFi98}, or rent-or-buy network
design~\cite{AwAzBa04,Umboh15,FeWiLe16}.

There are two major differences between BRP and the problems listed above.
First, these problems typically maintain some configuration of servers or
bought infrastructure and upon a new request (whose cost typically depends on
the distance to the infrastructure), decide about its reconfiguration (e.g.,
server movement or purchasing additional links). In contrast, in our model,
\emph{both} end-points of a communication request are subject to optimization.
Second, in the BRP problem a request reveals only very limited information
about the optimal configuration to serve it: There exist relatively long
sequences of requests that can be served with zero cost from a fixed
configuration. Not only can the set of such configurations be very large, but
such configurations may also differ significantly from each other.

Since the initial conference publication of this paper~\cite{disc16}, two
relaxations of the model were considered. First, Avin et al.~\cite{computing18}
studied a variant where requests are chosen randomly according to a probability
distribution fixed by an adversary. Under some additional assumptions, they gave
an algorithm which achieves logarithmic competitive ratio with high probability.
Second, Henzinger et al.~\cite{sigmetrics19learn} considered a~``learning''
variant, where requests correspond to edges of a graph that can be perfectly
partitioned (without inter-cluster edges). Note that for such setting, there
exists a static partitioning of zero cost, and thus the goal of an algorithm is
to ``learn'' such partitioning. The authors presented a distributed online
algorithm whose cost is asymptotically almost optimal, and show how to apply
this solution to a distributed union-find problem.


\section{A Simple Upper Bound}
\label{sec:upper}

As a warm-up and to present the model, we start with a straightforward $O(k^2
\cdot \ell^2)$-competitive deterministic algorithm \DET. At any time, \DET
serves a request, adjusts its internal structures (defined below)
accordingly and then possibly migrates some nodes. \DET operates in phases, and each
phase is analyzed separately. The first phase starts with the first request.

In a single phase, \DET maintains a helper structure: a complete graph on all
$\ell \cdot k$ nodes, with an edge present between each pair of nodes. We say
that a communication request is \emph{paid} (by \DET) if it occurs between
nodes from different clusters, and thus entails a cost for \DET. For each edge
between nodes $x$ and $y$, we define its weight~$w(x,y)$ to be the number of
paid communication requests between $x$ and~$y$ since the beginning of~the~current phase.

Whenever an edge weight reaches $\alpha$, it is called \emph{saturated}. If a
request causes the corresponding edge to~become saturated,
\DET computes a new placement of nodes (potentially for all of them), so that all
saturated edges are inside clusters (there is only one new saturated edge). If
this is not possible, node positions are not changed, the~current phase ends
with the current request, and a new phase begins with the next request. Note
that all edge weights are reset to zero at the beginning of a phase.

\begin{theorem}
\DET is $O(k^2 \cdot \ell^2)$-competitive.
\end{theorem}

\begin{proof}
We bound the costs of \DET and \OPT in a single phase. First, observe that
whenever an~edge weight reaches $\alpha$, its endpoint nodes will be collocated 
until the end of the phase, and therefore its weight is not
incremented anymore. Hence the weight of any edge is at most $\alpha$.

Second, observe that the graph induced by saturated edges always constitutes 
a~forest. Suppose that, at a time $t$,
two nodes $x$ and~$y$, which are not
connected by a saturated edge, become connected by a path of saturated edges.
From that time onward, \DET stores them in a single cluster. Hence, the
weight~$w(x,y)$ cannot increase at subsequent time points, and $(x,y)$ may
not become saturated. The forest property implies that the number of saturated
edges is smaller than $k \cdot \ell$.

The two observations above allow us to bound the cost of \DET in a single
phase. The number of reorganizations is at most the number of saturated edges,
i.e., at most~$k \cdot \ell$. As the cost associated with a single
reorganization is $O(k \cdot \ell \cdot \alpha)$, the total cost of all node
migrations in a single phase is at most $O(k^2 \cdot \ell^2 \cdot \alpha)$.
The communication cost itself is equal to the total weight of all edges, and
by the first observation, it is at most $\binom{k \cdot \ell}{2}
\cdot \alpha < k^2 \cdot \ell^2 \cdot \alpha$. Hence, for any phase $P$ (also
for the last one), it holds that $\DET(P) = O(k^2 \cdot \ell^2 \cdot \alpha)$.

Now we lower-bound the cost of \OPT on any phase $P$ but the last one. If \OPT
performs a node swap in $P$, it pays $2 \cdot \alpha$. Otherwise its assignment of
nodes to clusters is fixed throughout $P$. Recall that at the end of $P$, \DET
failed to reorganize the nodes. This means that for any static mapping of the
nodes to clusters (in particular the one chosen by \OPT), there is a
saturated inter-cluster edge. The communication cost over such an~edge incurred
by \OPT is at least $\alpha$ (it can be also strictly greater than~$\alpha$ as
the edge weight only counts the communication requests paid by \DET).

Therefore, the $\DET$-to-$\OPT$ cost ratio in any phase but the last one is at
most $O(k^2 \cdot \ell^2)$ and the cost of \DET on the last phase is at
most $O(k^2 \cdot \ell^2 \cdot \alpha)$. Hence,
$\DET(\sigma) \leq O(k^2 \cdot \ell^2) \cdot \OPT(\sigma) + O(k^2 \cdot
\ell^2 \cdot \alpha)$ for any input $\sigma$.
\end{proof}


\section{Algorithm {\sc Crep}}
\label{sec:crep}

In this section, we present the main result of this paper, a
\emph{Component-based REPartitioning algorithm (\CREP)} which achieves a
competitive ratio of $O((1 + 1/\eps) \cdot k \log k)$ with augmentation
$2+\eps$, for any~$\eps \geq 1/k$ (i.e., the augmented cluster
is of size at least $2k+1$). For technical convenience, we assume that 
$\eps \leq 2$. This assumption is without loss of generality: if the augmentation 
is $2+\eps > 4$, \CREP~simply uses each cluster only up to capacity $4k$.

\CREP maintains a similar graph structure as the
simple deterministic $O(k^2 \cdot \ell^2)$-competitive algorithm \DET from the
previous section, i.e., it keeps counters denoting how many times it paid for a
communication between two nodes. Similarly, at any time~$t$,
\CREP serves the current request, adjusts its internal structures, and then
possibly migrates nodes. Unlike \DET, however, the execution of \CREP is
\emph{not} partitioned into global phases: the reset of counters to zero can
occur at different times.

\subsection{Algorithm Definition}

We describe the construction of \CREP in two stages. The first stage uses
an~intermediate concept of \emph{communication components}, which are groups of at
most $k$ nodes. In the second stage, we show how components are assigned to
clusters, so that all nodes from any single component are always stored in a
single cluster.

\subsubsection{Stage 1: Maintaining Components}

Roughly speaking, nodes are grouped into components if they communicated a lot
recently. At the very beginning, each node is in a singleton component. Once
the cumulative communication cost between nodes distributed across $s$
components exceeds $\alpha \cdot (s-1)$, $\CREP$ merges them into a~single
component. If a resulting component size exceeds $k$, it becomes split
into singleton components.

More precisely, the algorithm maintains a time-varying \emph{partition of all
nodes into components}. As a helper structure, \CREP keeps a complete graph on
all $k \cdot \ell$ nodes, with an edge present between each pair of nodes. For
each edge between nodes $x$ and~$y$, \CREP maintains its weight $w(x,y)$. We
say that a communication request is \emph{paid} (by \CREP) if it occurs
between nodes from different clusters, and thus entails a~cost for \CREP. If
$x$ and $y$ belong to the same component, then $w(x,y) = 0$. Otherwise,
$w(x,y)$ is equal to the number of paid communication requests between $x$
and~$y$ since the last time when they were placed in \emph{different components} 
by \CREP. It is worth emphasizing that during an~execution of \CREP, it is
possible that $w(x,y) > 0$ even when $x$ and $y$ belong to~the~same cluster.

For any subset of components $\set = \{ C_1, C_2, \ldots, C_{|\set|} \}$ (called
\emph{component-set}), by~$w(\set)$ we denote the total weight of all edges
between nodes of $\set$. Note that positive weight edges occur only between
different components of~$\set$. We call a component-set \emph{trivial} if it
contains only one component; $w(\set) = 0$ in this case.

Initially, all components are singleton components and all edge weights are
zero. At time $t$, upon a~communication request between a pair of nodes $x$
and $y$, if $x$ and $y$ lie in the same cluster, the corresponding cost is~$0$
and \CREP does nothing. Otherwise, the cost entailed to \CREP is $1$, nodes
$x$ and $y$ lie in different clusters (and hence also in different
components), and the following updates of weights and components are
performed.

\begin{enumerate}

\item \emph{Weight increment.} Weight $w(x,y)$ is incremented.

\item \emph{Merge actions.} We say that a non-trivial component-set $\set = \{
C_1, \ldots, C_{|\set|} \}$ is \emph{mergeable} if $w(\set) \geq
(|\set|-1) \cdot \alpha$. If a mergeable component-set $S$ exists, then all its
components are merged into a single one. If multiple mergeable component-sets
exist, \CREP picks the one with maximum number of components, breaking ties
arbitrarily. Weights of all intra-$\set$ edges are reset to zero, and thus
intra-component edge weights are always zero. A mergeable set $\set$ induces
a~sequence of $|\set|-1$ \emph{merge actions}:
\CREP iteratively replaces two arbitrary components 
from~$\set$ by a component being their union (this constitutes a single merge
action).

\item \emph{Split action.} If the final component resulting from merge action(s)
has more than $k$ nodes, it is split into singleton
components. Note that weights of edges between these singleton components are
all zero as they have been reset by the preceding merge actions.

\end{enumerate}

We say that merge actions are \emph{real} if they are not followed
by a split action (at the same time point) and \emph{artificial} otherwise.

\subsubsection{Stage 2: Assigning Components to Clusters}

At time $t$, \CREP processes a communication request and recomputes components
as described in the first stage. Recall that we require that nodes of a single
component are always stored in a~single cluster. To maintain this property for
artificial merge actions, no actual migration is necessary. The property may
however be violated by real merge actions. Hence, in the following, we assume
that in the first stage \CREP found a mergeable component set $\set = \{ C_1, 
\ldots, C_{|\set|} \}$ that triggers $|\set|-1$ merge actions not 
followed by a split action.

\CREP consecutively processes each real merge action by migrating some nodes.
We describe this process for a single real merge action involving two
components $C_x$ and $C_y$. As a split action was not executed, $|C_x| +
|C_y| \leq k$, where $|C|$ denotes the number of component $C$ nodes.
Without loss of generality, $|C_x| \leq |C_y|$.

We may assume that $C_x$ and $C_y$ are in different clusters as otherwise
\CREP does nothing. If the cluster containing $C_y$ has $|C_x|$ free space,
then $C_x$ is migrated to this cluster. Otherwise, \CREP finds a cluster that
has at~most $k$ nodes, and moves both $C_x$ and $C_y$ there. We call the
corresponding actions \emph{component migrations}. By an~averaging argument,
there always exists a cluster that has at most $k$ nodes, and hence, with
$(2+\eps)$-augmentation, component migrations are always feasible.

\subsection{Analysis: Structural Properties}

We start with a structural property of components and edge weights.
The property states that immediately after \CREP merges (and
possibly splits) a component-set, no other component-set is~mergeable. This
property holds independently of the actual placement of components in
particular clusters.

\begin{lemma}
\label{lem:wS_bound}
At any time $t$, after \CREP performs all its actions,
$w(\set) < \alpha \cdot (|\set|-1)$ for any non-trivial component-set $\set$.
\end{lemma}

\begin{proof}
We prove the lemma by an induction on steps. Clearly, the lemma holds before an
input sequence starts as then $w(\set) = 0 \leq \alpha - 1 < \alpha \cdot
(|\set|-1)$ for any non-trivial set $\set$. We assume that it holds at time $t-1$
and show it for time $t$.

At time $t$, only a single weight, say $w(x,y)$, may be incremented. If after
the increment, \CREP does not merge any component, then clearly $w(\set) < \alpha
\cdot (|\set|-1)$ for any non-trivial set $\set$. Otherwise, at time $t$, \CREP
merges a~component-set $\A$ into a new component $C_\A$, and then possibly
splits $C_\A$ into singleton components. We show that
the lemma statement holds then for any non-trivial component-set $\set$. We
consider three cases.

\begin{enumerate}

\item Component-sets $\A$ and $\set$ do not share any common node. Then, $\A$ and
$\set$ consist only of components that were present already right before time~$t$
and they are all disjoint. The edge~$(x,y)$ involved in~communication at time
$t$ is contained in~$\A$, and hence does not contribute to the weight of
$w(\set)$. By the inductive assumption, the inequality 
$w(\set) < \alpha \cdot (|\set|-1)$ held right
before time $t$. As $w(\set)$ is not affected by \CREP's actions at step $t$, the
inequality holds also right after time $t$.

\item \CREP does not split $C_\A$ and $C_\A \in \set$. Let $\X = \set \setminus
\{C_\A\}$. Let $w(\A,\X)$ denote the total weight of all edges with one endpoint
in $\A$ and another in~$\X$. Recall that
\CREP always merges a~mergeable component-set with maximum number of components. 
As \CREP merged component-set~$\A$ and did not merge
(larger) component-set $\A \uplus \X$, $\A$ was mergeable ($w(\A) \geq \alpha \cdot
(|\A|-1)$), while $\A \uplus \X$ was not, i.e., $w(\A) + w(\A,\X) + w(\X) = w(\A
\uplus \X) < \alpha \cdot (|\A|+|\X|-1)$. Therefore, $w(\A,\X) + w(\X) < \alpha
\cdot |\X|$ right after weight $w(x,y)$ is incremented at time~$t$. Observe
that when component-set $\A$ is merged and all intra-$\A$ edges have their weights 
reset to zero, neither $w(\A,\X)$ nor $w(\X)$ is affected.
Therefore after \CREP merges $\A$ into $C_\A$, $w(\set) =
w(\A,\X) + w(\X) < \alpha \cdot |\X| = \alpha \cdot (|\set| - 1)$.

\item \CREP splits $C_\A$ into singleton components $B_1, B_2, \ldots, B_r$ 
and some of these components belong to~set $\set$. This time, we define $\X$ to be
the subset of $\set$ not containing these components ($\X$ might be also an empty set). In~the
same way as in the previous case, we may show that $w(\A,\X) + w(\X) < \alpha
\cdot |\X|$ after \CREP performs all its operations at time $t$. Hence, at this
time $w(\set) \leq w(\A,\X) + w(\X) < \alpha \cdot |\X| \leq \alpha \cdot (|\set|-1)$.
The last inequality follows as $\set$ has strictly more components than $\X$.
\qed
\end{enumerate}
\end{proof}

Since only one request is given at a time, and since all weights and $\alpha$
are integers, \cref{lem:wS_bound} immediately implies the following
result.

\begin{corollary}
\label{cor:mergeable_sets} Fix any time $t$ and consider weights right after
they are updated by \CREP, but before \CREP performs merge actions. Then,
$w(\set) \leq (|\set|-1) \cdot \alpha$ for any component-set $\set$. In particular,
$w(\set) = (|\set|-1) \cdot \alpha$ for a mergeable component-set~$\set$.
\end{corollary}

\subsection{Analysis: Overview}

In the remaining part of the analysis, we fix an~input sequence $\sigma$ and
consider a set $\spl(\sigma)$ of all components that are split by \CREP, i.e.,
components that were created by merge actions, but because of their size they
were immediately split into singleton components. Our goal is to compare both
the cost of \OPT and \CREP to $\sum_{C \in \spl(\sigma)} |C|$. 
Below, we provide the main intuitions for our approach.

For each component $C \in \spl(\sigma)$, we may track the history of how it was
created. This history corresponds to a tree $\T(C)$ whose root is $C$, the leaves are the singleton components containing nodes of $C$, and the internal nodes correspond
to components that are created by merging their children. Note that for any
two sets in $\spl(\sigma)$ their trees contain disjoint subsets of components.
Hence, for any $C \in \spl(\sigma)$, we want to relate the costs of \OPT and
\CREP due to processing, 
to the requests related to the components of $\T(C)$. (Some
components may not belong to any tree, but the related cost can be universally
bounded by a constant independent of input $\sigma$.)

In \cref{sec:opt_lower}, we lower-bound the cost of \OPT. Assume first that \OPT
does not migrate nodes. Fix any component $C \in \spl(\sigma)$. As its size is
greater than $k$, it spans $\Omega(|C|/k)$ clusters in the solution of \OPT. Note
that \cref{cor:mergeable_sets} lower-bounds the number of requests between
siblings in $\T(C)$. Then, for any assignment of nodes of $C$ to 
the clusters, $\Omega(|C| \cdot \alpha /k)$ requests are between clusters.
Additionally, if \OPT migrates nodes, then the amount of
request-related cost that \OPT saves, is dominated by the migration cost. In total,
the cost of \OPT related to~$\T(C)$ is at least $\Omega(|C| \cdot
\alpha / k)$.

In \cref{sec:crep_upper} and \cref{sec:crep_ratio}, we upper-bound the cost of
\CREP. Its request cost is asymptotically dominated by its migration cost, and
hence it is sufficient to bound the latter. If \CREP was always able to migrate
the smaller component to the cluster holding the larger component, then the
total migration cost related to components from $\T(C)$ could be bounded by
$\Omega(|C| \cdot \alpha \cdot \log k)$. (This bound is easy to observe when
$\T(C)$ is a fully balanced binary tree and all merged components are of equal
size.) Unfortunately, \CREP may sometimes need to migrate both components. 
However, if such migrations are expensive, then the
distribution of nodes in clusters becomes significantly more even. Consequently, 
the cost of expensive migrations can be charged to the cost of other migrations,
at the expense of an~extra $O(1 + 1/\eps)$ factor in the cost.
In total, the (amortized) cost of \CREP related to $\T(C)$ is at most
$\Omega((1+1/\eps) \cdot |C| \cdot \alpha \cdot \log k)$.

Finally, comparing bounds on \CREP and \OPT yields the desired bound on the
competitive ratio.

\subsection{Analysis: Lower Bound on OPT}
\label{sec:opt_lower}

In our analysis, we conceptually
replace any swap of two nodes performed by \OPT into two migrations of~the~corresponding nodes.

For any component $C$ maintained by \CREP, let $\tau(C)$ be the time of its
creation. A~non-singleton component~$C$ is created at $\tau(C)$ by the merge
of a component-set, henceforth denoted by~$\set(C)$. For a singleton component,
$\tau(C)$ is the time when the component that previously contained the sole
node of $C$ was split; $\tau(C) = 0$ if $C$ existed at the beginning of
input~$\sigma$. We use time $0$ as an artificial time point that occurred
before an actual input sequence.

For a non-singleton component $C$, we 
define $\F(C)$ as the set of the following (node, time) pairs:
\begin{align*}
	\F(C) = \biguplus_{B \in \set(C)} \{ B \}  \times \left\{ \tau(B)+1, \ldots, \tau(C) \right\}.
\end{align*}
Intuitively, $\F(C)$ tracks the history of all nodes of $C$ from the time 
(exclusively) they started belonging to some previous component $B$, until the time 
(inclusively) they become members of $C$. Note that for any two components 
$C_1,C_2$, sets~$\F(C_1)$ and $\F(C_2)$ are disjoint.
The union of all $\F(C)$ (over all components $C$) 
cover all possible node-time pairs (except for time zero). 

For a given component $C$, we say that a communication request between nodes
$x$ and $y$ at time~$t$ \emph{is contained} in $\F(C)$ if both $(x,t) \in \F(C)$
and $(y,t) \in \F(C)$. Note that only the requests contained in~$\F(C)$
could contribute towards later creation of~$C$ by \CREP. In fact, by
\cref{cor:mergeable_sets}, the number of these requests that
entailed an~actual cost to~\CREP is exactly $(|\set(C)| - 1) \cdot \alpha$.

We say that a migration of node $x$ performed by \OPT at time $t$ \emph{is
contained} in~$\F(C)$ if $(x,t) \in \F(C)$. For any component $C$, we define
$\OPT(C)$ as the cost incurred by \OPT due to requests contained in~$\F(C)$, 
plus the cost of~\OPT migrations contained in~$\F(C)$. The total cost of \OPT
can then be lower-bounded by the sum of $\OPT(C)$ over all components $C$.
(The cost of \OPT can be larger as $\sum_C \OPT(C)$ does not account for 
communication requests not contained in $\F(C)$ for any component~$C$.)

\begin{lemma}
\label{lem:merge_action_cut}
Fix any component $C$ and partition $\set(C)$ into a set of $g \geq 2$ disjoint
component-sets $\set_1, \set_2, \ldots, \set_g$. The number of communication requests
in $\F(C)$ that are between sets $\set_i$ is at least $(g-1) \cdot \alpha$.
\end{lemma}

\begin{proof}
Let $w$ be the weight measured right after its increment at time $\tau(C)$.
Observe that the number of all communication requests from $\F(C)$ that were
between sets $\set_i$ and that \emph{were paid} by \CREP is $w(\set(C)) -
\sum_{i=1}^g w(\set_i)$. It suffices to show that this amount is at least $(g-1)
\cdot \alpha$. By \cref{cor:mergeable_sets}, $w(\set(C)) = (|\set(C)|-1)
\cdot \alpha$ and $w(\set_i) \leq (|\set_i|-1) \cdot \alpha$. Therefore, $w(\set(C)) -
\sum_{i=1}^g w(\set_i) \geq (|\set(C)|-1) \cdot \alpha - \sum_{i=1}^g (|\set_i|-1)
\cdot \alpha = (g-1) \cdot \alpha$.
\end{proof}

For any component $C$ maintained by \CREP, let $Y_C$ denote the set of clusters
containing nodes of $C$ in the solution of \OPT after \OPT performs its
migrations (if any) at time~$\tau(C)$. In particular, if $\tau(C) = 0$, then
$Y_C$ consists of only one cluster that contained the sole node of $C$ 
at the beginning of an input sequence.

\begin{lemma}
\label{lem:opt_recursive_bound}
For any non-trivial component $C$, it holds that $\OPT(C) \geq (|Y_C| - 1)
\cdot \alpha - \sum_{B \in \set(C)} (|Y_B| - 1) \cdot \alpha$.
\end{lemma}

\begin{proof}
Fix a component $B \in \set(C)$ and any node $x \in B$. Let $\optmig(x)$ be the
number of \OPT migrations of~node $x$ at times $t \in \{ \tau(B)+1, \ldots,
\tau(C) \}$. Furthermore, let $Y'_x$ be the set of clusters that
contained $x$ at some moment of a~time $t \in \{ \tau(B)+1, \ldots, \tau(C)
\}$ (in the solution of \OPT). We extend these notions to components:
$\optmig(B) = \sum_{x \in B} \optmig(x)$ and $Y'_B = \bigcup_{x \in B} Y'_x$.
Observe that $|Y'_B| \leq |Y_B| + \optmig(B)$.
We say that $Y'_B$ are the clusters that were \emph{touched} by component $B$ 
(in the solution of \OPT).

By \cref{cor:mergeable_sets}, the number of communication requests between
components of~$\set(C)$ is $(|\set(C)|-1) \cdot \alpha$. However, $\OPT(C)$
includes the cost only due to these requests that are between different clusters.
Hence, to lower-bound $\OPT(C)$, we aggregate components of~$\set(C)$ into
component-sets called \emph{bundles}, so that any two bundles have their nodes
always in disjoint clusters. This way, any communication between nodes from
different bundles incurs a cost to \OPT.

The bundles with the desired property can be created by a natural iterative 
process. We start from $|\set(C)|$ bundles, each containing just a single 
component from~$\set(C)$. Then, we iterate over all clusters touched 
by any component of $\set(C)$, i.e., over all clusters from $\bigcup_{B \in \set(C)} Y'_B$.
For each such cluster $V$, let $H_V$ be the set 
of all components of $\set(C)$ that touched $V$. We then aggregate all bundles 
containing any component from $H_V$ into a single bundle.

On the basis of this construction, we may lower-bound the number of 
bundles. Initially, we have $|\set(C)|$ bundles. When we process a cluster $V \in 
\bigcup_{B \in \set(C)} Y'_B$, we aggregate at most $|H_V|$ bundles, and thus
the total number of bundles drops at most by $|H_V|-1$. Therefore, the final number of bundles 
is 
\begin{align*}
	p \,
	\geq &\; \textstyle |\set(C)| - \sum_{V \in \bigcup_{B \in \set(C)} Y'_B} (|H_V|-1) \\
	= &\; \textstyle |\bigcup_{B \in \set(C)} Y'_B| + \textstyle |\set(C)| - \sum_{V \in \bigcup_{B \in \set(C)} Y'_B} |H_V| \\
	= &\; \textstyle |\bigcup_{B \in \set(C)} Y'_B| + \textstyle |\set(C)| - \sum_{B \in \set(C)} |Y'_B| \\
	= &\; \textstyle |\bigcup_{B \in \set(C)} Y'_B| - \sum_{B \in \set(C)} (|Y'_B| - 1) \\
	\geq &\; \textstyle |Y_C| - \sum_{B \in \set(C)} (|Y'_B| - 1) \\
	\geq &\; \textstyle |Y_C| - \sum_{B \in \set(C)} (|Y_B|-1) - \sum_{B \in \set(C)} \optmig(B),
\end{align*}
where the second inequality follows as $Y_C \subseteq \bigcup_{B \in \set(C)} Y'_B$.

By \cref{lem:merge_action_cut}, the number of communication requests in
$\F(C)$ that are between different bundles is at least $(p-1) \cdot \alpha$,
and each of these requests is paid by \OPT.
Additionally, $\OPT(C)$ involves $\sum_{B \in \set(C)}
\optmig(B)$ node migrations in $\F(C)$, and therefore $\OPT(C) \geq (p-1) \cdot
\alpha + \sum_{B \in \set(C)} \optmig(B) \cdot \alpha
\geq (|Y_C|-1) \cdot \alpha - \sum_{B \in \set(C)} (|Y_B|-1) \cdot \alpha$.
\end{proof}

\begin{lemma}
\label{lem:opt_lower_bound}
For any input $\sigma$, $\OPT(\sigma) \geq \sum_{C \in \spl(\sigma)}
|C| / (2k) \cdot \alpha$.
\end{lemma}

\begin{proof}
Fix any component $C \in \spl(\sigma)$. Recall that $\T(C)$ is a tree
describing how $C$ was created: the leaves of $\T(C)$ are singleton
components containing nodes of $C$, the root is $C$ itself, and each internal
node corresponds to a~component created at a~specific time, by merging its
children.

We now sum $\OPT(B)$ over all components $B$ from $\T(C)$, including 
the root $C$ and the leaves $L(\T(C))$. The~lower bound given 
by \cref{lem:opt_recursive_bound} sums telescopically, i.e.,
\begin{align*}
	\textstyle \sum_{B \in \T(C)} \OPT(B) \;
		\geq &\; \textstyle (|Y_C|-1) \cdot \alpha - \sum_{B \in L(\T(C))} (|Y_B|-1) \cdot \alpha \\
		= &\; \textstyle (|Y_C|-1) \cdot \alpha,
\end{align*}
where the equality follows as any $B \in L(\T(C))$ is a singleton component,
and therefore $|Y_B| = 1$. As $C$ has $|C|$ nodes, it has to span at least
$\lceil |C|/k \rceil$ clusters of \OPT, and therefore $\sum_{B \in \T(C)}
\OPT(B) \geq (\lceil |C|/k \rceil -1) \cdot \alpha \geq |C| / (2 k) \cdot
\alpha$, where the second inequality follows because $C \in \spl(\sigma)$ and
thus $|C| > k$.

The proof is concluded by observing that, for any two components $C_1$
and~$C_2$ from~$\spl(\sigma)$, the corresponding trees~$\T(C_1)$ and $\T(C_2)$ do not share common
components, and therefore $\OPT(\sigma) \geq \sum_{C \in \spl(\sigma)} \sum_{B
\in \T(C)} \OPT(B) \geq \sum_{C \in \spl(\sigma)} |C| / (2k) \cdot \alpha$.
\end{proof}

\subsection{Analysis: Upper Bound on CREP}
\label{sec:crep_upper}

To bound the cost of \CREP, we fix any input $\sigma$ and introduce the following 
notions. Let~$M(\sigma)$ be the sequence of merge actions 
(real and artificial ones) performed by \CREP. 
For any real merge action $m \in M(\sigma)$, by
$\size(m)$ we denote the size of the smaller component that was
merged. For an~artificial merge action, we set $\size(m) = 0$.

Let $\final(\sigma)$ be the set of all components that exist when \CREP finishes
sequence~$\sigma$. Note that $w(\final(\sigma))$ is the total weight of all
edges after processing $\sigma$. We split 
$\CREP(\sigma)$ into two parts: the cost of serving requests, $\CREPreq(\sigma)$, 
and the cost of node migrations, $\CREPmig(\sigma)$. 

\begin{lemma}
\label{lem:crep_req}
For any input $\sigma$, $\CREPreq(\sigma) = |M(\sigma)| \cdot \alpha + w(\final(\sigma))$.
\end{lemma}

\begin{proof}
The proof follows by an induction on all requests of $\sigma$. Whenever \CREP
pays for the communication request, the corresponding edge weight is incremented
and both sides increase by $1$. At a time when $s$ components are merged, $s-1$
merge actions are executed and, by \cref{cor:mergeable_sets}, the sum of all
edge weights decreases exactly by $(s-1) \cdot \alpha$. Then, the value of both
sides remain unchanged.
\end{proof}

\begin{lemma}
\label{lem:crep_mig}
For any input $\sigma$, with $(2+\eps)$-augmentation, 
$\CREPmig(\sigma) \leq (1+4/\eps) \cdot \alpha \cdot \sum_{m \in M(\sigma)} \size(m)$.
\end{lemma}

\begin{proof}
If $\CREP$ has more than $2 k$ nodes in cluster $V_i$ (for $i \in
\{1,\ldots,\ell\}$), then we call the excess $|V_i| - 2 k$ the \emph{overflow} of $V_i$;
otherwise, the overflow of $V_i$ is zero. We denote the overflow of cluster
$V_i$ measured right after processing sequence~$\sigma$ by $\ovr^\sigma(V_i)$.
It is sufficient to show the following relation for any sequence $\sigma$:
\begin{equation}
\label{eq:crep_migrations}
 \CREPmig(\sigma) \;+\; \sum_{j=1}^\ell (4/\eps) \cdot \alpha \cdot \ovr^\sigma(V_j) 
 \; \leq \; (1+4/\eps)  \cdot \alpha \cdot \sum_{m \in M(\sigma)} \size(m).
\end{equation}
As the second summand of \eqref{eq:crep_migrations} is always non-negative,
\eqref{eq:crep_migrations} will imply the lemma. 
In other words, the lemma will be shown using amortized analysis, where
the amount $\sum_{j=1}^\ell (4/\eps) \cdot \alpha \cdot \ovr^\sigma(V_j)$ serves 
as a potential function.

The proof of \eqref{eq:crep_migrations} follows by an induction on all
requests in $\sigma$. Clearly, \eqref{eq:crep_migrations} holds trivially at the
beginning, as there are no overflows, and thus both sides of
\eqref{eq:crep_migrations} are zero. Assume that \eqref{eq:crep_migrations}
holds for a~sequence $\sigma$ and we show it for sequence $\sigma \cup \{ r \}$,
where $r$ is some request.

We may focus on a request $r$ which triggers the migration of
component(s), as otherwise
\eqref{eq:crep_migrations} holds trivially. Such a migration is triggered by a
real merge action $m$ of two components $C_x$ and $C_y$. We assume that $|C_x|
\leq |C_y|$, and hence $\size(m) = |C_x|$. Note that $|C_x| + |C_y| \leq k$,
as otherwise the resulting component would be split and no migration would
be performed.

Let $V_x$ and $V_y$ denote the cluster that held components $C_x$ and $C_y$,
respectively, and $V_z$ be the destination cluster for $C_x$ and $C_y$ (it is
possible that $V_z = V_y$). For any cluster~$V$, we denote the change in 
its overflow by $\Delta \ovr(V) = \ovr^{\sigma \cup \{r\}}(V) - \ovr^\sigma(V)$. 
It suffices to show that the
change of the left hand side of \eqref{eq:crep_migrations} is at most
the~increase of its right hand side, i.e.,
\begin{equation}
\label{eq:crep_mig_inc}
\CREPmig(r) \;+\; \sum_{V \in \{ V_x, V_y, V_z \}} (4/\eps) \cdot \alpha \cdot \Delta\ovr(V) 
\;\leq\; (1+4/\eps) \cdot |C_x| \cdot \alpha.
\end{equation}
For the proof, we consider three cases.

\begin{enumerate}
\item 
$V_y$ had at least $|C_x|$ empty slots. In this case, \CREP simply migrates
$C_x$ to $V_y$ paying $|C_x| \cdot \alpha$. Then, $\Delta \ovr(V_x) \leq 0$,
$\Delta \ovr(V_y) \leq |C_x|$, $V_z = V_y$, and thus \eqref{eq:crep_mig_inc}
follows.

\item 
$V_y$ had less than $|C_x|$ empty slots and $|C_y| \leq (2/\eps) \cdot |C_x|$.
\CREP migrates both $C_x$ and $C_y$ to~component $V_z$ and the incurred cost is 
$\CREPmig(r) = (|C_x| + |C_y|) \cdot \alpha \leq (1+2/\eps) \cdot |C_x| \cdot \alpha
< (1 + 4/\eps) \cdot |C_x| \cdot \alpha$. 
It remains to show that the second summand of \eqref{eq:crep_mig_inc} is~at~most $0$. 
Clearly, $\Delta \ovr(V_x) \leq 0$ and $\Delta \ovr(V_y) \leq 0$. 
Furthermore, the number of 
nodes in $V_z$ was at~most~$k$ before the migration by the definition of \CREP,
and thus is at~most $k + |C_x| + |C_y| \leq 2k$ after the migration.
This implies that $\Delta \ovr(V_z) = 0 - 0 = 0$.

\item 
$V_y$ had less than $|C_x|$ empty slots and $|C_y| > (2/\eps) \cdot |C_x|$.
As in the previous case, \CREP migrates $C_x$ and $C_y$ to~component $V_z$, 
paying $\CREPmig(r) = (|C_x| + |C_y|) \cdot \alpha < 2 \cdot |C_y| \cdot \alpha$.
This time, $\CREPmig(r)$ can be much larger than the right hand side 
of~\eqref{eq:crep_mig_inc}, and thus we resort to showing that 
the second summand~of~\eqref{eq:crep_mig_inc} is at most $ - 2 \cdot |C_y| \cdot \alpha$.

As in the previous case, $\Delta \ovr(V_x) \leq 0$ and $\Delta \ovr(V_z) = 0$. 
Observe that $|C_x| < (\eps / 2) \cdot |C_y| \leq (\eps / 2) \cdot k$.
As the migration of $C_x$ to $V_y$ was not possible, the initial number
of nodes in $V_y$ was greater than $(2 + \eps) \cdot k - |C_x| \geq (2+\eps/2) \cdot k$,
i.e., $\ovr^\sigma(V_y) \geq (\eps/2) \cdot k \geq (\eps/2) \cdot |C_y|$. 
As component $C_y$ was migrated out of $V_y$, the number of overflow nodes in $V_y$ changes by
\[
	\Delta \ovr(V_y) 
		= - \min \left\{ \,\ovr^\sigma(V_y), \,|C_y| \,\right\}
		\leq - (\eps/2) \cdot |C_y|.
\]
In the inequality above, we used $\eps \leq 2$.
Therefore, the second summand of~\eqref{eq:crep_mig_inc} is at most 
$(4/\eps) \cdot \alpha \cdot \Delta\ovr(V_y) \leq - 
(4/\eps) \cdot \alpha \cdot (\eps/2) \cdot |C_y|
= - 2 \cdot |C_y| \cdot \alpha$ as desired.
\end{enumerate}
\end{proof}


\subsection{Analysis: Competitive Ratio}
\label{sec:crep_ratio}

In the previous two subsections, we related $\OPT(\sigma)$ to the total
size~of components that are split by \CREP
(cf.~\cref{lem:opt_lower_bound}) and $\CREP(\sigma)$ to $\sum_{m \in
M(\sigma)} \size(m)$, where the latter amount is related to the merging
actions performed by \CREP (cf.~\cref{lem:crep_mig}). Now we link
these two amounts. Note that each split action corresponds to preceding
merge actions that led to the creation of the split component.

\begin{lemma}
\label{lem:bounding_merges}
For any input $\sigma$, it holds that 
$\sum_{m \in M(\sigma)} \size(m) 
	\leq \sum_{C \in \spl(\sigma)} |C| \cdot \log k +
	\sum_{C \in \final(\sigma)} |C| \cdot \log |C|$,
where all logarithms are binary.
\end{lemma}

\begin{proof}
We prove the lemma by an induction on all requests of $\sigma$. At the very
beginning, both sides of the lemma inequality are zero, and hence the induction
basis holds trivially. We assume that the lemma inequality is preserved for a
sequence $\sigma$ and we show it for sequence $\sigma \cup \{ r \}$, where $r$
is an arbitrary request. We may assume that $r$ triggers some merge actions,
otherwise the claim follows trivially.

First, assume $r$ triggered a sequence of real merge actions. We show that the
lemma inequality is preserved after processing each merge action. Consider a
merge action and let $C_x$ and $C_y$ be the components that are merged, with
sizes $p = |C_x|$ and $q = |C_y|$, where $p \leq q$ without loss of generality.
Due to the merge action, the right hand side of the lemma inequality increases
by
\begin{align*}
  (p + q) \cdot \log (p + q) & - p \cdot \log p - q \cdot \log q \\
		= &\; p \cdot (\log (p+q) - \log p) + q \cdot (\log (p+q) - \log q) \\
		\geq &\; p \cdot \log (p+q) / p \\
		\geq &\; p \cdot \log 2 = p.
\end{align*}
As~the~left hand side of the inequality changes exactly by $p$, the inductive
hypothesis holds.

Second, assume $r$ triggered a sequence of artificial merge actions (i.e., followed by a
split action) and let $C_1, C_2, \ldots, C_g$ denote components that were
merged to create a component $C$ that was immediately split. Then, the right
hand side of the lemma inequality changes by $- \sum_{i = 1}^g |C_i| \cdot
\log |C_i| + |C| \cdot \log k
\geq - \sum_{i = 1}^g |C_i| \cdot \log k + |C| \cdot \log k = 0$.
As~the~left hand side of the lemma inequality is unaffected by artificial
merge actions, the inductive hypothesis follows also in this case.
\end{proof}

\begin{theorem}
With augmentation at least $2+\eps$, \CREP is~$O((1 + 1/\eps) \cdot k \cdot \log k)$-competitive.
\end{theorem}

\begin{proof}
Fix any input sequence $\sigma$. 
By \cref{lem:crep_req} and \cref{lem:crep_mig}, 
\begin{align*}
	\CREP(\sigma) 
	= &\; \CREPmig(\sigma) + \CREPreq(\sigma) \\
	\leq &\; \textstyle (1+4/\eps) \cdot \alpha \cdot \sum_{m \in M(\sigma)} \size(m) + |M(\sigma)| \cdot \alpha + w(\final(\sigma)).
\end{align*}

Regarding a bound for $|M(\sigma)|$, we observe the following. First, if \CREP
executes artificial merge actions, then they are immediately followed by a
split action of the resulting component $C$. The number of artificial merge
actions is clearly at most $|C|-1 \leq |C|$, and thus the total number of all
artificial actions in $M(\sigma)$ is at most $\sum_{C \in \spl(\sigma)} |C|$.
Second, if $\CREP$ executes a real merge action $m$, then
$\size(m) \geq 1$. Combining these two bounds yields $|M(\sigma)| \leq \sum_{m
\in M(\sigma)} \size(m) + \sum_{C \in \spl(\sigma)} |C|$. We use this inequality 
and later apply \cref{lem:bounding_merges} to~bound $\sum_{m \in M(\sigma)} 
\size(m)$ obtaining
\begin{align*}
	& \CREP(\sigma) - w(\final(\sigma)) \\
  &\quad \leq  \textstyle (1+4/\eps) \cdot \alpha \cdot \sum_{m \in M(\sigma)} \size(m) 
    + |M(\sigma)| \cdot \alpha \\ 
	&\quad \leq \textstyle (2+4/\eps) \cdot \alpha \cdot \sum_{m \in M(\sigma)} \size(m) 
		+ \alpha \cdot \sum_{C \in \spl(\sigma)} |C|  \\
	&\quad \leq \textstyle (2+4/\eps) \cdot \alpha \cdot \left( 
			\sum_{C \in \spl(\sigma)} |C| \cdot \log k + \sum_{C \in \final(\sigma)} |C| \cdot \log |C|
			\right)
			+ \alpha \cdot \sum_{C \in \spl(\sigma)} |C| \\
	&\quad \leq \textstyle (3+4/\eps) \cdot \alpha \cdot 
			\sum_{C \in \spl(\sigma)} |C| \cdot \log k 
			+ (2+4/\eps) \cdot \alpha \cdot
			\sum_{C \in \final(\sigma)} |C| \cdot \log |C|.
\end{align*}
By \cref{lem:opt_lower_bound}, $\sum_{C \in \spl(\sigma)} |C| \cdot \alpha \leq 
2k \cdot \OPT(\sigma)$. This yields
\begin{align*}
	\CREP(\sigma)
	\leq &\; \textstyle O(1+1/\eps) \cdot k \cdot \log k \cdot \OPT(\sigma) + \beta,
\intertext{where}
	\beta = &\; O(1+1/\eps) \cdot \alpha \cdot
			\sum_{C \in \final(\sigma)} |C| \cdot \log |C|
			+ w(\final(\sigma)).
\end{align*}
To bound $\beta$, observe that the component-set $\final(\sigma)$
contains at most $k \cdot \ell$ components, and hence
by~\cref{lem:wS_bound}, $w(\final(\sigma)) < k \cdot \ell \cdot
\alpha$. Furthermore, the maximum of $\sum_{C \in \final(\sigma)} |C| \cdot
\log |C|$ is achieved when all nodes in a specific cluster constitute a single
component. Thus, $\sum_{C \in \final(\sigma)} |C| \cdot \log |C|
\leq \ell \cdot ((2+\eps) \cdot k) \cdot \log ((2+\eps) \cdot k) = O(\ell
\cdot k \cdot \log k)$.
In total, $\beta = O((1 + 1/\eps) \cdot \alpha \cdot \ell \cdot k \cdot \log k)$, 
i.e., it can be upper-bounded by a constant independent of input sequence $\sigma$,
which concludes the proof.
\end{proof}


\section{Online Rematching}
\label{sec:k-two}

Let us now consider the special case where clusters are of size two ($k=2$,
arbitrary~$\ell$). This can be viewed as an online maximal (re)matching problem:
clusters of size two contain (``match'') exactly one pair of nodes, and
maximizing pairwise communication within each cluster is equivalent to
minimizing inter-cluster communication.

\subsection{Greedy Algorithm}

We define a natural greedy online algorithm~$\GREEDY$, parameterized by a real
positive number $\lambda$. Similarly to our other algorithms,
\GREEDY  maintains an edge weight for each pair of nodes. 
The weights of all edges are initially zero. Weights of intra-cluster edges
are always zero and weights of inter-cluster edges are related to the number
of paid communication requests between edge endpoints. 

When facing an inter-cluster request between nodes $x$
and~$y$, $\GREEDY$ increments the weight $w(e)$, where $e = (x,y)$. Let $x'$
and $y'$ be the nodes collocated with $x$ and~$y$, respectively. If after the
weight increase, it holds that $w(x,y) + w(x',y') \geq \lambda
\cdot \alpha$, $\GREEDY$ performs a swap: it places $x$ and $y$ in one
cluster and $x'$ and~$y'$ in another; afterwards, it resets the weights of
edges $(x,y)$ and $(x',y')$ to 0. Finally, \GREEDY pays for the request
between $x$ and $y$. Note that if the request triggered a migration, then
\GREEDY does not pay its communication cost.

\subsection{Analysis}

We use $E$ to denote the set of all edges.
Let $\MGREEDY$ ($\MOPT$) denote the set of all edges $e = (u,v)$, such 
that $u$ and $v$ are collocated by \GREEDY (\OPT). 
Note that $\MGREEDY$ and $\MOPT$ are perfect matchings on the set of all nodes.

To estimate the total cost of \GREEDY, we use amortized analysis with 
an appropriately defined potential function. First, 
we associate the following edge-potential with any edge $e$:
\[
	\Phi(e) = \begin{cases}
		0 			& \textrm{if $e \in \MGREEDY$,} \\
		- w(e) 	& \textrm{if $e \in \MOPT \setminus \MGREEDY$,} \\
		f \cdot w(e) & \textrm{if $e \notin \MOPT$ and $e \notin \MGREEDY$,} \\
	\end{cases}
\]
where $f \geq 0$ is a constant that will be defined later. 

The union of $\MGREEDY$ and $\MOPT$ constitutes a set of alternating cycles:
an alternating cycle of length~$2 j$ (for some $j \geq 1$) consists of $2 j$
nodes, $j$ edges from $\MGREEDY$ and $j$~edges from $\MOPT$, interleaved. The
case $j = 1$ is~degenerate: such a cycle consists of a single edge from $\MGREEDY
\cap \MOPT$, but we still count it as a cycle of length $2$. 
It turns out that the number of these alternating cycles is a good measure of 
similarity between matchings of \GREEDY and \OPT (when these matchings are 
equal, the number of cycles is maximized). We define the
cycle-potential as
\[
	\Psi = - g \cdot K \cdot \alpha,
\]
where $K$ is the number of all alternating cycles and $g \geq 0$ is a constant that will
be defined later.

To simplify the analysis, we slightly modify the way weights are increased by
\GREEDY. The modification is~applied only when the weight increment triggers 
a~node migration. Recall that this happens when there is an inter-cluster
request between nodes $x$ and $y$. The corresponding weight $w(x,y)$ is then
increased by $1$. After the~increase, it holds that $w(x,y) + w(x',y') \geq
\lambda \cdot \alpha$. (Nodes $x'$ and $y'$ are those collocated with $x$ and $y$,
respectively.) Instead, we increase $w(x,y)$ possibly by a~smaller amount, so
that $w(x,y) + w(x',y')$ becomes \emph{equal} to $\lambda \cdot \alpha$. This
modification allows for a more streamlined analysis and is local: before and
after the modification, \GREEDY performs a migration and right after that, 
it resets weight $w(x,y)$ to zero.

We split the processing of a communication request $(x,y)$ into three stages. In
the first stage, \OPT performs an~arbitrary number of migrations. In the
second stage, the weight $w(x,y)$ is increased accordingly and both \OPT and
\GREEDY serve the request. It is possible that the weight increase triggers a
node swap of \GREEDY, in which case its serving cost is zero. Finally, in the
third stage, \GREEDY may perform a node swap.

We show that for an appropriate choice of $\lambda$, $f$ and $g$, for all
three stages described above the following inequality holds:
\begin{equation}
\label{eq:rematch_bound}
	\textstyle \Delta \GREEDY + \Delta \Psi + \sum_{e \in E} \Delta \Phi(e) \leq 7 \cdot \Delta \OPT.
\end{equation}
Here, $\Delta \GREEDY$ and $\Delta \OPT$ denote the increases of \GREEDY's and
\OPT's cost, respectively. $\Delta \Psi$ and $\Delta \Phi(e)$ are the changes
of the potentials $\Psi$ and $\Phi(e)$. The 7-competitiveness then immediately
follows from summing \eqref{eq:rematch_bound} and bounding the initial and
final values of the potentials. 

\begin{lemma}
\label{lem:opt_swap}
If $2 \cdot (f+1) \cdot \lambda + g \leq 14$, then \eqref{eq:rematch_bound}
holds for the first stage.
\end{lemma}

\begin{proof}
We consider any node swap performed by \OPT. Clearly, for such an event
$\Delta \GREEDY = 0$ and $\Delta \OPT = 2 \cdot \alpha$. The number of cycles
decreases at most by one, and thus $\Delta \Psi \leq g \cdot \alpha$.

We now upper-bound the change in the edge-potentials. Let $\eold_1$ and
$\eold_2$ be the edges that were removed from $\MOPT$ by the swap and let
$\enew_1$ and $\enew_2$ be the edges added to $\MOPT$. For any $i \in
\{1,2\}$, $\Delta \Phi(\enew_i) \leq 0$ as the initial value of
$\Phi(\enew_i)$ is at least $0$ and the final value of $\Phi(\enew_i)$ is at
most $0$. Similarly, $\Delta \Phi(\eold_i) \leq (f+1) \cdot w(\eold_i)$ as the
initial value of $\Phi(\eold_i)$ is at least $-w(\eold_i)$ and the final value
of $\Phi(\eold_i)$ is at most $f \cdot w(\eold_i)$.

Summing up, $\sum_{e \in E} \Delta \Phi \leq (f+1) \cdot (w(\eold_1) +
w(\eold_2)) \leq 2 \cdot (f+1) \cdot \lambda \cdot \alpha$ as the weight of each edge
is at most $\lambda \cdot \alpha$. By combining the bounds above and using the
lemma assumption, we obtain $\Delta \GREEDY + \sum_{e \in E} \Delta \Phi(e) +
\Delta \Psi \leq 0 + 2 \cdot (f+1) \cdot \lambda \cdot \alpha + g \cdot \alpha 
\leq 14 \cdot \alpha = 7 \cdot \Delta \OPT$.
\end{proof}

\begin{lemma}
\label{lem:rematch_req}
If $f \leq 6$, then \eqref{eq:rematch_bound} holds for the second stage.
\end{lemma}

\begin{proof}
In this stage, both \GREEDY and \OPT serve a communication request between
nodes $x$ and $y$. Let $e_c = (x,y)$. As neither \GREEDY nor \OPT migrates any
nodes in this stage, the structure of alternating cycles remains
unchanged, i.e., $\Delta \Psi = 0$. Furthermore, only edge~$e_c$ may change its
weight, and therefore, among all edges, only the edge-potential of $e_c$ may
change. We consider two cases.
\begin{enumerate}

\item If $e_c \in \MGREEDY$, then $\Delta \GREEDY = 0$ and $\Delta \OPT \geq 0$. As
$w(e_c)$ is unchanged, $\Delta \Phi(e_c) = 0$, and therefore 
$\Delta \GREEDY + \Delta \Phi(e_c) = 0 \leq \Delta \OPT$. 

\item If $e_c \notin \MGREEDY$, then let $\Delta w(e_c) \leq 1$ denote the increase 
of the weight of edge~$e_c$. Note that $\Delta \GREEDY \leq \Delta w(e_c)$: 
either no migration is triggered and $\Delta \GREEDY = \Delta w(e_c) = 1$
or a migration is triggered and then \GREEDY does not pay for the request.

If $e_c \in \MOPT$, then $\Delta \OPT = 0$ and $\Delta \Phi(e_c) = -w(e_c)$.
Thus, $\Delta \GREEDY + \Delta \Phi(e_c) \leq 0 = \Delta \OPT$. Otherwise, 
$e_c \notin \MOPT$, in which case 
$\Delta \OPT = 1$. Furthermore, $\Delta \Phi(e_c) = f \cdot \Delta w(e_c)$,
and thus $\Delta \GREEDY + \Delta \Phi(e_c) = (f+1) \cdot \Delta w(e_c) \leq f+1 =
(f+1) \cdot \Delta \OPT$.
\end{enumerate}

Therefore, in the second stage, $\Delta \GREEDY + \Delta \Psi + \sum_{e \in E}
\Delta \Phi(e) \leq (f+1) \cdot \Delta \OPT$,
which implies \eqref{eq:rematch_bound} as we assumed $f \leq 6$.
\end{proof}

\begin{lemma}
\label{lem:greedy_swap}
If $2 + \lambda \leq g \leq f \cdot \lambda - 2$, 
then  \eqref{eq:rematch_bound} holds for the third stage.
\end{lemma}

\begin{proof}
In the third stage (if it is present), $\GREEDY$ performs a swap. Clearly, for
such an event $\Delta \GREEDY = 2 \cdot \alpha$ and $\Delta \OPT = 0$. 

There are four edges involved in a swap: let $(x,x')$ and $(y,y')$ be the
edges that were in $\MGREEDY$ before the swap and let $(x,y)$ and $(y,y')$ be
the new edges in $\MGREEDY$ after the swap. Note that $w(x,x') = w(y,y') = 0$
before and after the swap. By the definition of \GREEDY and our modification
of weight updates, $w(x,y) + w(x',y') = \lambda
\cdot \alpha$ before the swap, and after the swap these weights are reset to
zero.

For any edge $e$, let $\winit(e)$ and $\Phiinit(e)$ denote the weight and the
edge-potential of~$e$ right before the swap. By~the~bounds above, $\Delta
\GREEDY + \sum_{e \in E} \Delta \Phi(e) + \Delta\Psi = 2 \cdot \alpha -
\Phiinit(x,y) - \Phiinit(x',y') + \Delta \Psi$, and hence to show
\eqref{eq:rematch_bound} it suffices to show that the latter amount is at most
$7 \cdot \Delta \OPT = 0$. We consider three cases.

\begin{figure}
\centering
\includegraphics[width=0.95\textwidth]{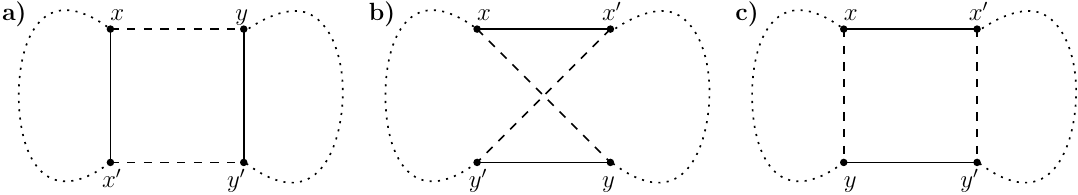}
\caption{Three cases in the analysis of the third stage (a swap performed by
\GREEDY). Solid edges denote edges that were removed from $\MGREEDY$ because
of the swap, dashed ones denote the ones that were added to $\MGREEDY$. Dotted
paths denote the remaining parts of the involved alternating cycle(s).}
\label{fig:rematch}
\end{figure}

\begin{enumerate}

\item
Assume that edges $(x,x')$ and $(y,y')$ were in different alternating
cycles before the swap, see \cref{fig:rematch}a. Then the number
of~alternating cycles decreases by one, and hence $\Delta \Psi = g \cdot
\alpha$. Let $C$ be the cycle that contained edge~$(x,x')$. Node $x$ is adjacent
to an edge from $C$ that belongs to $\MOPT$. (It is possible that this edge is
$(x,x')$; this occurs in the degenerate case when $C$ is of length~$2$.) As
$\MOPT$ is a matching, $(x,y) \notin \MOPT$. Analogously, $(x',y') \notin
\MOPT$. Therefore, $\Phiinit(x,y) + \Phiinit(x',y') = f \cdot w(x,y) + f \cdot
w(x',y') = f \cdot \lambda \cdot \alpha$. Using the lemma assumption,
$\Delta \GREEDY + \sum_{e \in E} \Delta \Phi(e) + \Delta\Psi = (2 + g - f
\cdot \lambda) \cdot \alpha \leq 0$.

\item
Assume that edges $(x,x')$ and $(y,y')$ belonged to the same cycle and
it contained the nodes in the order $x,x',\ldots,y,y',\ldots$,
see~\cref{fig:rematch}b. In this case, it holds that 
$\Delta \Psi = 0$, since 
the number of alternating cycles is unaffected by the swap. By similar
reasoning as in the previous case, neither $(x,y)$ nor $(x',y')$ belong to
$\MOPT$, and thus again, $\Phiinit(x,y) + \Phiinit(x',y') = f \cdot w(x,y) + f
\cdot w(x',y') = f \cdot \lambda \cdot \alpha$. In this case, $\Delta \GREEDY
+ \sum_{e \in E} \Delta \Phi(e) + \Delta\Psi = (2 - f \cdot \lambda) \cdot
\alpha \leq (2 + g - f \cdot \lambda) \cdot \alpha \leq 0$.

\item 
Assume that edges $(x,x')$ and $(y,y')$ belonged to the same cycle and
it contained the nodes in the order $x,x',\ldots,y',y,\ldots$,
see~\cref{fig:rematch}c. When the swap is performed, the number of
alternating cycles decreases, and thus $\Delta \Psi = -g \cdot \alpha$. Unlike
the previous cases, here it is possible that $(x,y)$ belonged to $\MOPT$.
(This happens when $x$ and $y$ were adjacent on the alternating cycle.)
Similarly, it is possible that $(x',y') \in \MOPT$. 
But even in such a case, we may lower-bound the initial values of the
corresponding edge-potentials: $\Phiinit(x,y) + \Phiinit(x',y') \geq -
\winit(x,y) - \winit(x',y') = - \lambda \cdot \alpha$. Using the lemma
assumption, $\Delta \GREEDY + \sum_{e \in E} \Delta \Phi(e) + \Delta\Psi 
\leq (2 - g + \lambda) \cdot \alpha \leq 0$.
\end{enumerate}

\end{proof}

\begin{theorem}
\label{thm:rematching}
For $\lambda = 4/5$, $\GREEDY$ is $7$-competitive.
\end{theorem}

\begin{proof}
We choose $f = 6$ and $g = 14/5$. The chosen values of $\lambda$, $f$ and $g$
satisfy the conditions of \cref{lem:opt_swap}, \cref{lem:rematch_req},
and \cref{lem:greedy_swap}. Summing these
inequalities over all stages occurring while serving an input sequence~$\sigma$
yields
\begin{equation*}
	\textstyle \GREEDY(\sigma) + (\Psi_\textrm{final} - \Psi_\textrm{initial})
	+ \sum_{e \in E} \left( 
		\Phi_\textrm{final}(e) - \Phi_\textrm{initial}(e) \right)
	\leq 7 \cdot \OPT(\sigma),
\end{equation*}
where $\Psi_\textrm{final}$ and $\Phi_\textrm{final}(e)$ denote the final
values of the potentials and $\Psi_\textrm{initial}$ and
$\Phi_\textrm{initial}(e)$ their initial values. We observe that all the
potentials occurring in the inequality above are lower-bounded and
upper-bounded by values that are independent of the input sequence $\sigma$.
That is, $\Psi_\textrm{final} - \Psi_\textrm{initial} \geq - g \cdot \ell
\cdot \alpha$ (as the number of~alternating cycles is at most $\ell$) and
$\Phi_\textrm{final}(e) - \Phi_\textrm{initial}(e) \geq - (f+1) \cdot w(e)
\geq - (f+1) \cdot \lambda \cdot \alpha$ (as all edge weights are always 
at most $\lambda \cdot \alpha$). The number of edges is exactly $\binom{2
\cdot \ell}{2}$, and therefore
\begin{align*}
	 \GREEDY(\sigma) 
	\leq &\; 7 \cdot \OPT(\sigma) 
	+ \textstyle g \cdot \ell \cdot \alpha + \binom{2 \cdot \ell}{2} \cdot (f+1) \cdot 
	\lambda \cdot \alpha \\
	\leq &\; 7 \cdot \OPT(\sigma) 
	+ O(\ell^2 \cdot \alpha).
\end{align*}
This concludes the proof.
\end{proof}


\section{Lower Bounds}
\label{sec:lower}

In order to shed light on the optimality of the presented online algorithm, we
next investigate lower bounds on the competitive ratio achievable by any
(deterministic) online algorithm. We start by showing a reduction of the BRP
problem to online paging, which will imply that already for two clusters the
competitive ratio of the problem is at least $k-1$. We strengthen this bound,
providing a lower bound of $k$ that holds for any amount of augmentation, as
long as the augmentation does not allow to put all nodes in a single 
cluster. The proof uses the averaging argument. We refine this approach for a
special case of online rematching ($k = 2$ without augmentation), for which we
present a~lower bound of $3$.

\subsection{Lower Bound by Reduction to Online Paging}
\label{sec:paging}

\begin{theorem}
Fix any $k$. If there exists a $\gamma$-competitive deterministic algorithm $B$
for BRP for two clusters, each of~size~$k$, then there exists a
$\gamma$-competitive deterministic algorithm $P$ for the paging problem with 
cache size $k-1$ and where the number of different pages is $k$.
\end{theorem}

\begin{proof}
The pages are denoted by $p_1,p_2,\ldots,p_k$. Without loss of generality, we
assume that the initial cache is equal to $p_1,p_2,\ldots,p_{k-1}$. We fix any
input sequence $\sigma^P = (\sigma^P_1, \sigma^P_2, \sigma^P_3, \ldots)$ for the
paging problem, where $\sigma^P_t$ denotes the $t$-th accessed page. We show
how to construct, an online algorithm $P$ for the paging
problem that proceeds in the following online manner. 
The algorithm internally runs the algorithm~$B$, 
starting on the initial assignment of nodes to clusters that will be
defined below. For a requested page $\sigma^P_t$, it creates a subsequence
of~communication requests for the BRP problem, runs $B$ on them, and serves
$\sigma^P_t$ on the basis of $B$'s responses.

We use the following $2k$ nodes for the BRP problem: paging nodes $p_1,p_2,
\ldots, p_k$, auxiliary nodes $a_1,a_2,\ldots,$ $a_{k-1}$, and a special node
$s$. We say that the node clustering is \emph{well aligned} if one cluster
contains the node $s$ and $k-1$ paging nodes, and the other cluster contains
one paging node and all auxiliary nodes. There is a natural bijection between
possible cache contents and well aligned configurations: the cache consists of
the $k-1$ paging nodes that are in the same cluster as node $s$. (Without loss
of generality, we may assume that the cache of any paging algorithm is always
full, i.e., consists of $k-1$ pages.) If the configuration $c$ of a BRP
algorithm is well aligned, $\textsc{cache}(c)$ denotes the corresponding cache
contents.

The initial configuration for the BRP problem is the well aligned
configuration corresponding to the initial cache (pages
$p_1,p_2,\ldots,p_{k-1}$ in the cache).

For any paging node $p$, let $\comm(p)$ be a subsequence of communication
requests for the BRP problem, consisting of the request $(p, s)$, followed by
$\binom{k-1}{2}$ requests to all pairs of auxiliary nodes. Given an input
sequence $\sigma^P$ for online paging, we construct the input sequence
$\sigma^B$ for the BRP problem in the following way: For a request
$\sigma^P_t$, we repeat a~subsequence $\comm(\sigma^P_t)$ till the node
clustering maintained by $B$ becomes well aligned and $\sigma^P_t$ becomes
collocated with $s$. Note that $B$ must eventually achieve such a node
configuration: otherwise its cost would be arbitrarily large while a sequence
of repeated $\comm(\sigma^P_t)$ subsequences can be served at~a~constant
cost---the competitive ratio of~$B$ would then be unbounded. We denote the
resulting sequence of $\comm(\sigma^P_t)$ subsequences by
$\comm_t(\sigma^P_t)$.

To construct the response to the paging request $\sigma^P_t$, the algorithm
$P$ runs $B$ on $\comm_t(\sigma^P_t)$. Right after processing
$\comm_t(\sigma^P_t)$, the node configuration $c$ of $B$ is well aligned and
$\sigma^P_t$ is collocated with~$s$. Hence, $P$ may change its cache
configuration to $\textsc{cache}(c)$: such a response is feasible, because
since $\sigma^P_t$ is collocated with $s$, it is included by $P$ in the cache.
Furthermore, we may relate the cost of $P$ to the cost of~$B$: If $P$ modifies
the cache contents, the corresponding cost is $1$, as exactly one page has to
be fetched. Such a change occurs only if $B$ changed the clustering 
(at a~cost of at least $2 \cdot \alpha$). Therefore, $2
\cdot \alpha \cdot P(\sigma^P_t) \leq B(\comm_t(\sigma^P_t))$, which,
summed over all requests from sequence $\sigma^P$, yields $2 \cdot
\alpha \cdot P(\sigma^P) \leq B(\sigma^B)$.

Now we show that there exists an (offline) solution \OFF to $\sigma^B$, whose
cost is exactly $2 \cdot \alpha \cdot \OPT(\sigma^P)$. Recall that, for a
paging request $\sigma^P_t$, $\sigma^B$ contains the~corresponding sequence
$\comm_t(\sigma^P_t)$. Before serving the first request
of~$\comm_t(\sigma^P_t)$, \OFF changes its state to a well aligned
configuration corresponding to the cache of $\OPT$ right after serving paging
request $\sigma^P_t$. This ensures that the subsequence $\comm_t(\sigma^P_t)$
is free for \OFF. Furthermore, the cost of node migration of \OFF is $2 \cdot
\alpha$ (two paging nodes are swapped) if $\OPT$ performs a fetch, and 0 if
$\OPT$ does not change its cache contents. Therefore,
$\OFF(\comm_t(\sigma^P_t)) = 2 \cdot \alpha \cdot \OPT(\sigma^P_t)$, which
summed over the entire sequence $\sigma^P$ yields $\OFF(\sigma^B) = 2 \cdot
\alpha \cdot \OPT(\sigma^P)$.

As $B$ is $\rho$-competitive for the BRP problem, there exists a constant
$\beta$, such that for any sequence $\sigma^P$ and the corresponding sequence
$\sigma^B$, it holds that $B(\sigma^B) \leq \gamma \cdot \OFF(\sigma^B) +
\beta$. Combining this inequality with the inequalities between $P$ and $B$
and between \OFF and \OPT yields
\[
	2 \cdot \alpha \cdot P(\sigma^P) \leq 
	B(\sigma^B) \leq \gamma \cdot \OFF(\sigma^B) + \beta
	 =\gamma \cdot 2 \cdot \alpha \cdot \OPT(\sigma^P) + \beta,
\]
and therefore $P$ is $\gamma$-competitive.
\end{proof}

As any deterministic algorithm for the paging problem with cache size $k-1$
has a competitive ratio of~at~least $k-1$~\cite{SleTar85}, we obtain the
following result.

\begin{corollary}
The competitive ratio of the BRP problem on two clusters is at least $k-1$. 
\end{corollary}

\subsection{Additional Lower Bounds}
\label{sec:lower-bounds}

\begin{theorem}\label{thm:loweraugmk}
No~$\delta$-augmented deterministic online algorithm \ONL
can achieve a competitive ratio smaller than~$k$, as long as~$\delta < \ell~$.
\end{theorem}

\begin{proof}
In our construction, all nodes are numbered from $v_0$ to $v_{n-1}$. All
presented requests are edges in~a~ring graph on these nodes with edge $e_i$
defined as $(v_i,v_{(i+1) \mod n })$ for~$i = 0, \ldots, n-1$. At any time,
the adversary gives a communication request between an arbitrary pair of nodes
not collocated by \ONL. As $\delta < \ell$, \ONL cannot fit the entire ring in
a single cluster, and hence such a pair always exists. Such a request
entails a cost of~at~least~$1$ for \ONL. This way, we may define an~input
sequence $\sigma$ of~arbitrary length, such that $\ONL(\sigma) \geq
|\sigma|$.

Now we present $k$ offline algorithms $\OFF_1, \OFF_2, \ldots, \OFF_k$, such
that, neglecting an initial node reorganization they perform before
the input sequence starts, the sum of their total costs on $\sigma$ is exactly
$|\sigma|$. Toward this end, for any $j \in \{0,\ldots,k-1\}$, we define a
set~$\cut(j) = \{ e_j, e_{j+k}, e_{j+2k},
\ldots, e_{j+(\ell-1)\cdot k} \}$. For any~$j$, set $\cut(j)$ defines a
natural partitioning of all nodes into clusters, each containing $k$ nodes.
Before processing $\sigma$, the algorithm $\OFF_j$ first migrates its nodes
(paying at most $n \cdot \alpha$) to the clustering defined by $\cut(j)$ and
then never changes the node placement.

As all sets $\cut(j)$ are pairwise disjoint, for any request $\sigma_t$,
exactly one algorithm $\OFF_j$ pays for the request, and thus $\sum_{j=1}^k
\OFF_j(\sigma_t) = 1$. Therefore, taking the initial node reorganization into
account, we obtain that $\sum_{j=1}^k \OFF_j(\sigma) \leq k \cdot n \cdot
\alpha + \ONL(\sigma)$. By the averaging argument, there exists an offline
algorithm $\OFF_j$, such that $\OFF_j(\sigma) \leq \frac{1}{k} \cdot
\sum_{j=1}^k \OFF_j(\sigma) \leq n \cdot \alpha + \ONL(\sigma) / k$.
Thus, $\ONL(\sigma) \geq k \cdot \OFF_j(\sigma) - k \cdot n \cdot
\alpha \geq k \cdot \OPT(\sigma) - k \cdot n \cdot \alpha$.
The~theorem follows because the additive constant $k \cdot n \cdot \alpha$
becomes negligible as the length of $\sigma$ grows.
\end{proof}

\begin{theorem}
No deterministic online algorithm \ONL can achieve a competitive ratio 
smaller than~$3$ for the case $k = 2$ (without augmentation).
\end{theorem}

\begin{proof} As in the previous proof, we number the nodes from $v_0$ to
$v_{n-1}$. We distinguish three types of node clusterings. Configuration A:
$v_0$ collocated with $v_1$, $v_2$ collocated with $v_3$, other nodes
collocated arbitrarily; configuration B: $v_1$ collocated with $v_2$, $v_3$
collocated with $v_0$, other nodes collocated arbitrarily; configuration C:
all remaining clusterings.

Similarly to the proof of \cref{thm:loweraugmk}, the adversary always
requests a communication between two nodes not collocated by \ONL.
This time the exact choice of such nodes is relevant: \ONL receives request to
$(v_1,v_2)$ in configuration A, and to $(v_0,v_1)$ in configurations B and C.

We define three offline algorithms. They keep nodes
$\{v_0,\ldots,v_3\}$ in the first two clusters and the remaining nodes in the
remaining clusters (the remaining nodes never change their clusters). 
More concretely, $\OFF_1$ keeps nodes $\{v_0,\ldots,v_3\}$ always in
configuration A and $\OFF_2$ always in configuration B. Furthermore, we define
the third algorithm $\OFF_3$ that is in configuration B if \ONL is in
configuration A, and is in configuration A if \ONL is in configuration B or C.

We split the cost of $\ONL$ into the cost for serving requests, $\ONLreq$, and
the cost paid for its migrations, $\ONLmig$. Observe that, for~any~request
$\sigma_t$, $\OFF_1(\sigma_t) + \OFF_2(\sigma_t) = \ONLreq(\sigma_t)$.
Moreover, as $\OFF_3$ does not pay for any request and migrates only when 
\ONL does, $\OFF_3(\sigma_t) \leq \ONLmig(\sigma_t)$. Summing up,
$\sum_{j=1}^3 \OFF_j(\sigma_t) \leq \ONL(\sigma_t)$ for any request $\sigma_t$.
Taking into account the initial reconfiguration of nodes in $\OFF_j$ solutions
(which involves at~most one swap of cost $2 \cdot \alpha$), we obtain that
$\sum_{j=1}^3 \OFF_j(\sigma) \leq 2 \cdot \alpha + \ONL(\sigma)$. Hence, by the
averaging argument, there exists $j \in \{1,2,3\}$, such that $\ONL(\sigma)
\geq 3 \cdot \OFF_j(\sigma) - 2 \cdot \alpha \geq 3 \cdot \OPT(\sigma) - 
2 \cdot \alpha$. This concludes the proof, as $2 \cdot \alpha$ becomes
negligible as the length of $\sigma$ grows.
\end{proof}

\section{Conclusion}
\label{sec:conclusion}

This paper initiated the study of a natural dynamic partitioning problem which
finds applications, e.g., in the context of virtualized distributed systems
subject to changing communication patterns. We derived upper and lower bounds,
both for the general case as well as for a special case, related to a dynamic
matching problem. The natural research direction is to develop better
deterministic algorithms for the non-augmented variant of the general case,
improving over the straightforward $O(k^2 \cdot \ell^2)$-competitive algorithm
given in \cref{sec:upper}. While the linear dependency on $k$ is
inevitable (cf.~\cref{sec:lower}), it is not known whether
an~algorithm whose competitive ratio is independent of~$\ell$ is possible. We
resolved this issue for the $O(1)$-augmented variant, for which we gave an
$O(k \log k)$-competitive algorithm.

\bibliographystyle{siamplain}
\bibliography{references}

\end{document}